\documentclass[12pt,onecolumn,draftclsnofoot,]{IEEEtran}

\ifCLASSOPTIONcompsoc
  \usepackage[nocompress]{cite}
\else
  \usepackage{cite}
\fi

\usepackage{lscape}
\usepackage{etex}
\usepackage{graphicx}
\usepackage{multirow}
\usepackage{amsmath}
\usepackage{amssymb}
\usepackage{color}
\usepackage{tikz}
\usepackage{url}
\usepackage{placeins}
\usepackage{float}
\usepackage{multirow}
\usepackage{array}
\usepackage{footmisc}

\usepackage[bf,sc,small]{caption} 
\usepackage{listings}
\usepackage{algorithm, algpseudocode}
\algnewcommand{\Inputs}[1]{%
	\State \textbf{Inputs:}
	\Statex \hspace*{\algorithmicindent}\parbox[t]{.8\linewidth}{\raggedright #1}
}
\algnewcommand{\Output}[1]{%
	\State \textbf{Output:}
	\Statex \hspace*{\algorithmicindent}\parbox[t]{.8\linewidth}{\raggedright #1}
}
\algnewcommand{\Initialize}[1]{%
	\State \textbf{Initialize:}
	\Statex \hspace*{\algorithmicindent}\parbox[t]{.8\linewidth}{\raggedright #1}
}
\usepackage{layouts}
\usepackage{tabularx}
\usepackage{mathtools}
\usepackage{adjustbox}
\usepackage{xcolor}
\usepackage{enumitem}

\usepackage[caption=false]{subfig}
\usepackage{xspace}
\usepackage{changepage}
\usepackage{epsfig}
\usepackage{booktabs}
\usepackage{hyperref}
\usepackage{microtype}  	

\hypersetup{colorlinks=true, urlcolor=black, citecolor=blue, linkcolor = black}

\DeclareMathAlphabet{\mathitbf}{OML}{cmm}{b}{it}
\DeclareMathAlphabet{\mathbfit}{OML}{cmm}{b}{it}

\definecolor{mygreen}{rgb}{0,0.6,0}
\definecolor{myred}{rgb}{1,0.333,0.333}

\def\Tau{{\rm T}}
\usepackage{amsmath}
\usepackage{amsthm}
\usepackage{amssymb}
\usepackage{tikz}
\usepackage[all,cmtip]{xy}
\theoremstyle{definition}

\newtheorem{thm}{Theorem}
\newtheorem{cor}{Corollary}

\newtheorem{lem}{Lemma}

\newtheorem{define}{Definition}

\newcommand\autorefeq[1]{\hyperref[#1]{Equation~\eqref{#1}}}%
\newcommand\autorefapp[1]{\hyperref[#1]{Appendix~\ref{#1}}}%
\newcommand\autorefalg[1]{\hyperref[#1]{Algorithm~\ref{#1}}}%
\newcommand\autorefcor[1]{\hyperref[#1]{Corollary~\ref{#1}}}%
\newcommand\autorefprop[1]{\hyperref[#1]{Proposition~\ref{#1}}}%
\newcommand\autorefproperty[1]{\hyperref[#1]{Property~\ref{#1}}}%
\usepackage[utf8]{inputenc}
\usepackage{pict2e,picture}

\newsavebox\CBox
\newlength\CLength
\def\Circled#1{\sbox\CBox{#1}%
  \ifdim\wd\CBox>\ht\CBox \CLength=\wd\CBox\else\CLength=\ht\CBox\fi
    \makebox[1.2\CLength]{\makebox(0,1.2\CLength){\put(0,0){\circle{1.4\CLength}}}%
    \makebox(0,1.2\CLength){\put(-.5\wd\CBox,0){#1}}}}

\newif\ifanonymous

\begin{document}

\title{Multi-Purpose Aerial Drones for Network Coverage and Package Delivery}

\ifanonymous
\author{}
\institute{}
\else
\author{Mohammadjavad Khosravi,
	    Hamid Saeedi, 
        Hossein Pishro-Nik 
\IEEEcompsocitemizethanks{
	\IEEEcompsocthanksitem M. Khosravi and H. Pishro-Nik are with the Department of Electrical and Computer Engineering, University of Massachusetts, Amherst, MA,01003 USA\protect, E-mail: mkhosravi@umass.edu, pishro@engin.umass.edu.\protect
	\IEEEcompsocthanksitem H. Saeedi is with the School of Electrical and Computer Engineering, Tarbiat Modares University, Tehran, Iran, E-mail: hsaeedi@modares.ac.ir}}

\fi
\renewcommand\footnotemark{}

\maketitle


\vspace{-2cm}
\begin{abstract}
Unmanned aerial vehicles (UAVs) have become important in many applications including last-mile deliveries, surveillance and monitoring, and wireless networks. This paper aims to design UAV trajectories that simultaneously perform multiple tasks. We aim to design UAV trajectories that minimize package delivery time, and at the same time provide uniform coverage over a neighborhood area which is needed for applications such as network coverage or surveillance. We first consider multi-task UAVs for a simplified scenario where the neighborhood area is a circular region with the post office located at its center and the houses are assumed to be uniformly distributed on the circle boundary. We propose a trajectory process such that if according to which the drones move, a uniform coverage can be achieved while the delivery efficiency is still preserved. We then  consider a more practical scenario in which the delivery destinations are arbitrarily distributed in an arbitrarily-shaped region. We also do not assume any restrictions on the package arrivals. We show that simultaneous uniform coverage and efficient package delivery is possible for such realistic scenarios. This is shown using both rigorous analysis as well as simulations.

\end{abstract}

\vspace{-1cm}

\IEEEpeerreviewmaketitle

\section{Introduction}\label{sec:intro}
Commercial unmanned aerial vehicles (UAVs), commonly known as drones, deployed in an unmanned aerial system (UAS), have recently drawn increased interest from private industry and academia, owing to their autonomy, flexibility, and broad range of application domains. With the on-going miniaturization of sensors and processors and ubiquitous wireless connectivity, drones are finding many new uses in enhancing our way of life. Applications of UAV technology exist in agriculture \cite{barrientos2011aerial}, surveying land or infrastructure \cite{babel2017curvature,avellar2015multi,lin2014hierarchical}, security \cite{businessinsider,WRAL,usatoday,Savuran2016EfficientRP,wang2019development}, cinematography \cite{otto2018optimization}, health care \cite{Zipline,futurism,unmannedaerial} and emergency operations \cite{adams2011survey,Winn2014AnalysisOT,Nedjati2016CompleteCP,Raap2017TrajectoryOU}.



An important emerging application of drones is on-demand delivery of goods and services which is shown to be cost-competitive relative to traditional ground-based delivery methods  \cite{Scott2017DroneDM,otto2018optimization,Murray2015TheFS,Agatz2018OptimizationAF,Wang2017TheVR,tatham2017flying,nedjati2016post,Networksextended,rabta2018drone,chowdhury2017drones,DAndrea2014GuestEC,welch2015cost}. The drones can provide on-demand, inexpensive, and convenient access to the goods and items already in or near an urban area, including consumer goods, fast-food, medicine, and even on-demand groceries. In the design and scheduling of on-demand delivery application, the goal usually is to minimize the overall delivery time/distance \cite{ulmer2018same,otto2018optimization,Murray2015TheFS,Ha2018OnTM}. To this end, we can consider the delivery efficiency as the ratio of the actual distance traveled by the drones to the minimum feasible distance that needs to be traveled to take care of a set of package delivery jobs. The notion of efficiency will be made precise in Section \ref{sec:Ideal_case}.

Another important application of drones is their deployment in communications and surveillance \cite{motlagh2016low,ono2016wireless,mozaffari2016unmanned,lyu2016cyclical,chetlur2017downlink,zeng2016throughput,sun2019cooperative,mozaffari2015drone,lyu2016placement,kalantari2016number,mozaffari2016efficient,mozaffari2017mobile,wu2017joint,lyu2018uav,wu2018joint,lyu2016cyclical,chetlur2017downlink,zhang2019spectrum,enayati2019moving,Zuckerberg2014ConnectingTW,Sharma2016UAVAssistedHN,liu2019comp,bor2018spatial}. In the former case, the drones are also referred to as aerial base stations (ABS)
 \cite{chandrasekharan2016designing}. In many cases, the ABS's are assumed to be moving along some pre-designed trajectories \cite{enayati2019moving,mozaffari2015drone,mozaffari2016unmanned}. The latter case, referred to as surveillance drones (SD), is usually associated with the drones that can carry video cameras and transmit video to provide new perspectives in visual surveillance \cite{bonetto2015privacy}. Although these two applications may seem fundamentally different, they share a common requirement: they usually have to fly along trajectories so as to provide a relatively uniform coverage over the area on which they operate. Throughout this paper, such applications are referred to as uniform-coverage applications (UCA).

Since drones can be used in many applications, an interesting idea is to design UAS's that simultaneously perform multiple tasks. This could significantly improve the efficiency of such systems. In this paper, we aim to systematically investigate this idea for the first time. As a first step, we consider a residential region where drones are used as the last-mile delivery tools within the area. Since these drones are already flying all over the area and providing some kind of aerial coverage, we may want  to use them in a UCA framework.
If this is the case, an important question would be whether the same mobility patterns can provide a uniform coverage in the area of interest. Alternatively, if we modify the patterns to achieve a uniform coverage, do we necessarily have to lose anything in terms of delivery efficiency?

To get an insight into the proposed question, consider the 780-acre University of Massachusetts (UMASS) campus that contains about 170 buildings (Figure~\ref{fig:umassmap}) in which we assume that the last-mile delivery office is  located in the lower-left corner of the figure with 10 operating drones. The drones start flying in straight lines with constant velocity to deliver the package to the building of interest and fly back to the post office. It is not difficult to see that this is the most efficient delivery profile\footnote{It is easy to see that invoking any practical limitation such as safety considerations can only increase the travel distance. Moreover, it is worth noting that the details are not consequential here as the point being made is that normal operation of drones in straight-lines normally creates non-uniform coverage.}. We refer to this delivery algorithm as the "benchmark algorithm" throughout this paper. Now we investigate the coverage associated to this mobility pattern. To do so, we divide the maps into small regions and find the average number of drones on that region at an arbitrary time instant through a simulation setup. The results have been shown on a heat map in Figure~\ref{fig:heatmapumass}.

\begin{figure}[htbp]%
	\centering
	\subfloat[University of Massachusetts (UMASS) campus]{
		\includegraphics[width=0.45\columnwidth]{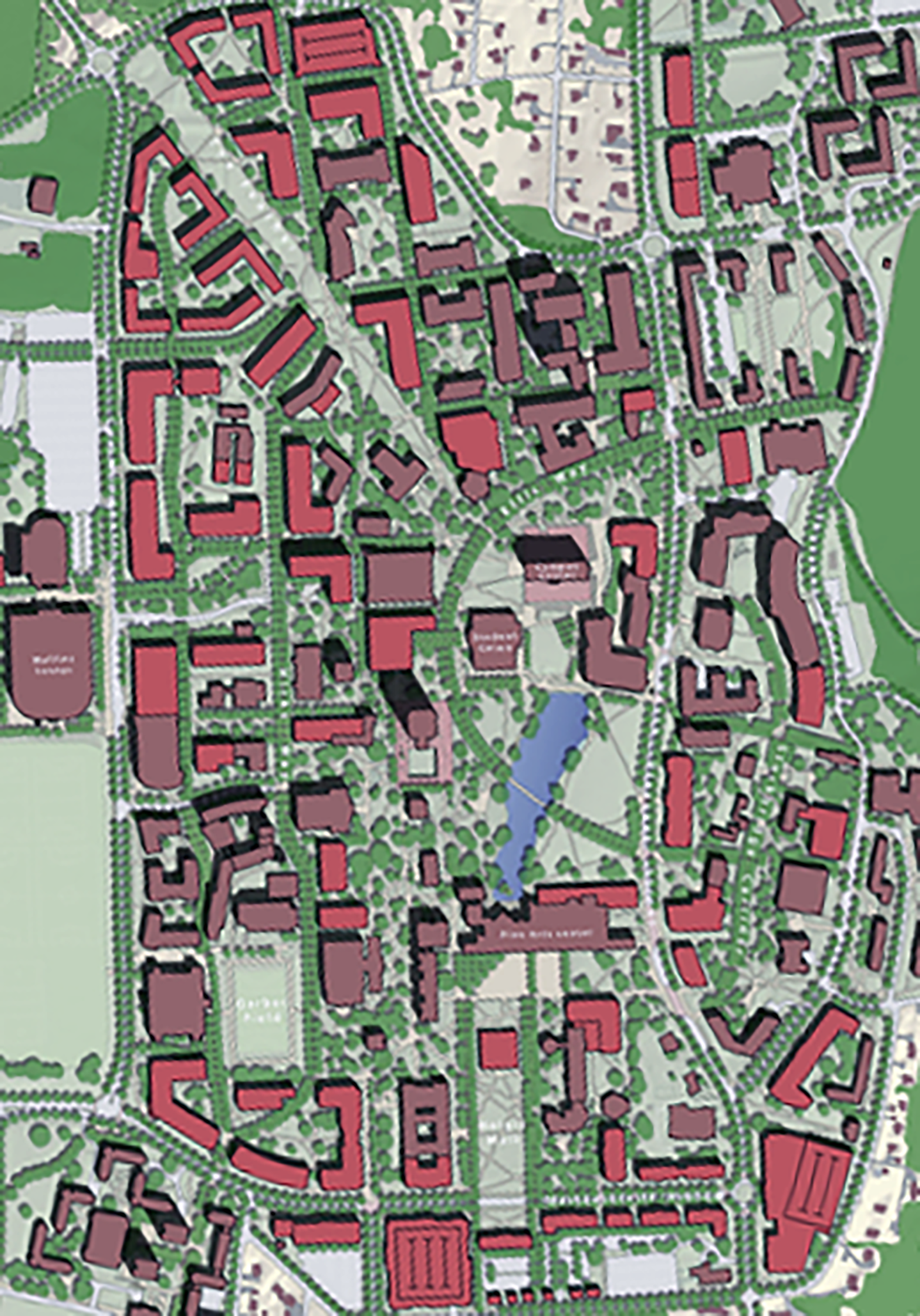}\label{fig:umassmap}}\hfill
	\subfloat[Heat-map of average number of drones for the \textit{fixed-speed-direct-line} algorithm]{
		\includegraphics[width=0.5\columnwidth]{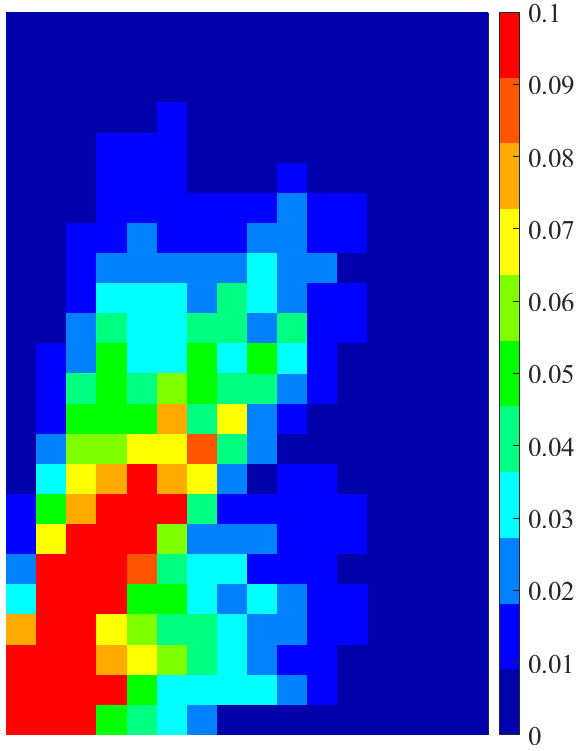}\label{fig:heatmapumass}}
	\caption{Multi-purpose drone algorithm for a residential area}%
	\label{fig:umassbenchmark}%
\end{figure}

As can be seen, the coverage is quite far from uniform which suggests that the idea of multi-purpose UAS may not actually work. Surprisingly, we will demonstrate that this is not the case. In this paper, we design efficient drone delivery systems that can simultaneously provide a fairly uniform coverage. This is achieved through designing mobility trajectories on which the drones move with variable speeds. We first consider a simplified scenario where we assume a circular region with the post office located at its center (referred to as the ideal case). The houses are assumed to be uniformly distributed on the circle boundary. Assuming the package arrivals are also uniform, we propose a trajectory process such that if according to which the drones move, a uniform coverage can be achieved while the delivery efficiency tends to 1. Next, we consider a more practical scenario in which the delivery destinations are arbitrarily distributed in an arbitrarily-shaped region. We also do not assume any restrictions on the package arrivals. In this case, we also show that simultaneous uniform coverage and efficient package delivery is practically possible.

The rest of the paper is organized as follows: in Section \ref{sec:Preliminaries}, we provide some definitions and discussions that are needed throughout the paper. In Section \ref{sec:Ideal_case}, we introduce our system model, scenario, our proposed algorithm for the ideal case (simplified scenario) and analytically prove the uniformity of the coverage and the efficiency of package delivery of our proposed algorithm. In Section \ref{sec:Practical_case}, we present the practical scenario, and after describing
the steps of our proposed algorithm, we prove the coverage uniformity. Section \ref{sec:conclude}, provides the simulation results, and Section V concludes our work.





\vspace{-5mm}
\section{Preliminaries}
\label{sec:Preliminaries}

\subsection{Binomial Point Processes}

If a fixed number of points are independently and identically distributed (i.i.d.) on a compact set $W \in R^d$, we say that these points can be modeled by general binomial point process (BPP) \cite{haenggi2012stochastic}. If these points are distributed uniformly within the same compact set, then we say the points are distributed according to a uniform BPP.

\subsection{Uniform Coverage}
We first need to clarify what we exactly mean by a uniform coverage. Uniform coverage can be considered from two perspectives: one is related to ensemble averages, and the other is related to time averages as discussed below.

In \cite{enayati2019moving}, authors obtain trajectories for UCAs according to ensemble averages. Specifically, they aim at designing trajectory processes for which, at any time snapshot, the locations of drones are distributed according to a uniform BPP process over the neighborhood area. This means that at any time $t$, the locations of drones are uniform and i.i.d. across the region. Here, the average is a taken over any sources of randomness in the scenario.

The other perspective is to look at the time averages. Roughly speaking, if we  divide the intended area to small equal cells, we can look at the percentage of the time each cell is covered over time and require that all the cells are covered equally over a long period of time.

Depending on the application, one of the above definitions might be more useful. Nevertheless, as it turns out, under mild conditions the trajectory processes can be made ergodic in the sense that both conditions can be satisfied simultaneously \cite{enayati2019moving}. In this paper, we consider the first definition (ensemble average view) for the ideal case in Section \ref{sec:Ideal_case}. This is because in that section, we make specific assumptions for probability distributions. On the other hand, in Section \ref{sec:Practical_case}, since we do not want to make any assumptions about probability distributions, we follow the second definition.

\subsection{Efficiency of package delivery}
Here, we make the notion of package delivery efficiency precisely.  Let ${\cal A (\mathrm{C})}$ be the set of all possible delivery algorithms satisfying the set of conditions and requirements $\mathrm{C}$. For example, for a given geometry, we could require that the algorithms are able to deliver $m$ arriving packages using $D$ drones with the average velocity $V_{avg}$ assuming each drone can carry only one package at a time. Since there is uncertainty and randomness in the operation (for example, the package destinations are not predetermined, and could follow a known or unknown statistical distribution), we need to consider a probabilistic view. More specifically, let the underlying probability space be represented as $(\Omega, \mathcal{F}, P)$. This probability space captures all non-deterministic aspect of the problem.

Consider an Algorithm $A \in \cal A (\mathrm{C})$. Let $\Tau_m(A)$ indicate the expected value of the time to deliver $m$ packages using Algorithm $A$, where the expectation is taken over the probability space $(\Omega, \mathcal{F}, P)$. Define $T_m^*$ as
$ T_m^*= \inf \{T_m(A): A \in \cal A (\mathrm{C})\}. $
Intuitively, $T_m^*$ provides the smallest average delivery time possible in a setting. This gives us a means to define package delivery efficiency for an any Algorithm $A \in \cal A (\mathrm{C})$.

\begin{define}
 Consider a set of delivery algorithms ${\cal A (\mathrm{C})}$ satisfying the set of conditions and requirements $\mathrm{C}$. We define the efficiency of the package delivery for an Algorithm $A \in \cal A (\mathrm{C})$ as follows
 \vspace{-5mm}
 \begin{equation}
\eta = \frac{T_m^*}{T_m(A)},      0 \leq \eta \leq 1
\end{equation}
\end{define}
If $\eta$ is close to $1$, it means that the algorithm is more efficient.

\section{Ideal case}
\label{sec:Ideal_case}

Here, we first explain the system model and scenario for ideal case. Next, we propose our algorithm which delivers the packages and provides the uniform coverage over the regions.

\vspace{-5mm}

\subsection{System model}
 Figure \ref{fig:Idealcasesneighborhood_area} shows the neighborhood area over which we want to provide the uniform coverage. We assume that $D$ drones deliver the arriving packages from the post office (at the center of region) to the $N$ destination houses and at the same time, they are used for a UCA. There are $N$ houses in the neighborhood area, which are destinations of the arrival packages. The houses are uniformly and independently distributed at the boundary of the circular region.  We assume packages are continuously arriving at the post office center. In other words, it is assumed that there are always packages in the post office to be delivered by the drones. Let $X_1, X_2, X_3, ...$ be the sequence of random variables that correspond to the sequence of incoming packages. More specifically, we say that the $i^{th}$ package must be delivered to the $k^{th}$ house, if $X_i= k$, where $k \in \{1, 2, ..., N\}$.

 \begin{figure}[htbp]%
 	\centering 	
 	\includegraphics[width=0.5\columnwidth]{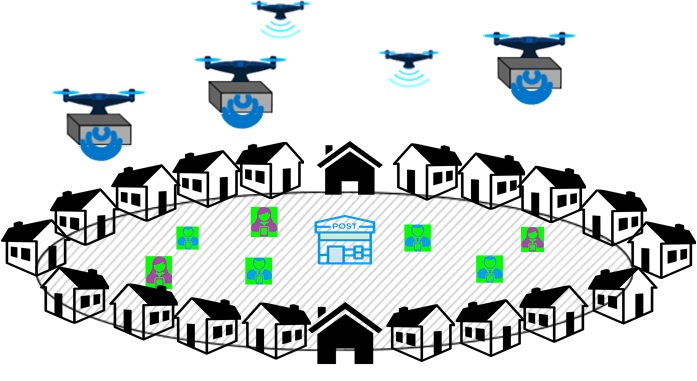} 	
 	\caption{Neighborhood area for Ideal case}
 	\label{fig:Idealcasesneighborhood_area} 	
 \end{figure}

To compare efficiency of different algorithms fairly, we assume that all the drones fly with the average velocity, i.e., $V_{avg}$. The average is computed over the running time of the delivery algorithm. The time needed for one drone to reach the neighborhood edge from the post office in a straight line by average velocity $V_{avg}$ is denoted by $\tau$, i.e., $\tau = \frac{\rho-\gamma}{V_{avg}}$ where $\gamma$ is the radius of the post office center, and $\rho$ is the radius of the entire neighborhood area. For simplicity, throughout the paper, we ignore the down times (i.e. as nights) and remove them from our analysis.
\vspace{-5mm}
\subsection{The Scenario for Ideal Case}
We assume that $D$ drones deliver the arriving packages from the post office in a circular neighborhood area. Figure \ref{fig:modelparam} shows the parameters of this scenario.  $\theta_i$ ($0 \leq \theta_i \leq \theta_{max}$) is the angle of the $i^{th}$ house on the perimeter of the circle sector. In case of  a full circle, $\theta_{max}$ is equal to $2\pi$ as  in Fig. \ref{fig:Idealcasesneighborhood_area}. The whole neighborhood area $A$ is defined as in (\ref{eq:targetarea}). We assume houses are distributed uniformly over the neighborhood edge. We also assume package destinations are uniformly distributed over $1, 2, ..., N$.

\begin{figure}[htbp]%
	\centering
	\includegraphics[width=0.5\columnwidth]{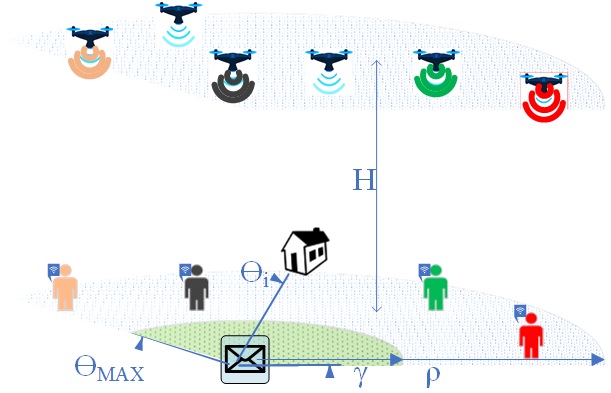}%
	\caption{Parameters of our system model}%
	\label{fig:modelparam}%
\end{figure}
\vspace{-7mm}

\begin{equation} \label{eq:targetarea}
A= \{(r, \theta): \gamma \leq r \leq \rho; 0 \leq \theta \leq \theta_{max} \}
\end{equation}

\subsection{Lower bound for $T_m^*$}
Here, we obtain a lower bound for $T_m^*$ for the ideal case.
\begin{lem} \label{lem:lem-low-Tau}
	In Ideal case, we have $ T_m^* \geq \frac{2m\tau}{D}$,
where $\tau$ and $D$ are defined above.
\end{lem}

\begin{proof}
    Let's first assume there is only one drone. For delivering any of the packages, the drone must travel a distance $d_i \geq 2(\rho-\gamma)$. Let $t_i$ be the time devoted to the delivery of the $i$th package. Then, the total time for delivery of $m$ packages will be at least $\sum_{i=1}^{m} t_i$ and the total distance traveled is $\sum_{i=1}^{m} d_i$. By assumption, the average speed is $V_{avg}$, therefore
    \[\sum_{i=1}^{m}t_i=\frac{\sum_{i=1}^{m} d_i}{V_{avg}} \geq \frac{2m(\rho-\gamma)}{V_{avg}}=2m \tau. \]

	Now, if there are $D$ drones, for simultaneously delivering $m$ packages, a minimum time of $\frac{2\tau m}{D}$ is necessary. Since this is true for all  $A \in \cal A (\mathrm{C})$, we conclude
\[ T_m^* \geq \frac{2m\tau}{D}.\]
\end{proof}


\subsection{The Algorithm}
\begin{algorithm}
	\small
	\begin{algorithmic}[1]
		\Function{DeliveryCost}{$D, m, X$}		
		\Inputs{
			D drones with average speed \textit{V}\\
			$m$ number of package to be delivered\\
			$X$ arrival packages which are distributed over $1, 2, ..., N$}
		\vspace{1cm}
		\Output{Total time to deliver $m$ packages ($T_m$)}
		\For{\texttt{<i=1; j<=$D$>}}
		\State Generate random variable $T_i$ uniform over $(0, \tau)$.
		\State Assign $i^{th}$ package to $i^{th}$ drone
		\State $i^{th}$ drone flies at $T_i$ over a straight line with $V_d(t)$ at angle $\theta_i$
		\EndFor
		\State $j=D+1$;
		\While{\texttt{j<=m>}}		
		\State Assign $j^{th}$ package to a free drone (say $i^{th}$ drone)
		\State $i^{th}$ drone flies right away over a straight line with $V_d(t)$ at angle $\theta_j$				
		\EndWhile
		\EndFunction
	\end{algorithmic}
	\caption{Algorithm corresponding to the ideal case}
	\label{alg:case1}	
\end{algorithm}
Here, we propose a multipurpose algorithm for the ideal case, i.e., an algorithm that can be used both for delivery of packages as well as uniform coverage. The simplifying assumptions of the ideal case makes the design of such algorithms very easy for this case. In fact, the main idea comes from properly randomizing the initial take-off times of the drones as well as properly choosing varying speeds for drones during delivery. In the proposed algorithm, referred to as Algorithm 1, first, we choose the take off times of drones, $T_1, T_2, ..., T_D$, independently and uniformly from $(0, \tau)$. A package $X_i = k (1 \leq k \leq N, i = 1, 2, ...)$ is assigned to a free drone to be delivered. Each drone first flies to a predetermined altitude of $H$, then flies in a straight line with angle $\theta_k$ (the direction of the destination) towards the neighborhood edge. 
When the drone reaches the neighborhood edge and delivers its assigned package, it returns to the origin on the same angle to complete the first cycle and this action repeats continuously. Figure~\ref{fig:idealcase} shows this trajectory process.

\begin{figure}[htbp]%
	\centering
	\includegraphics[width=0.5\columnwidth]{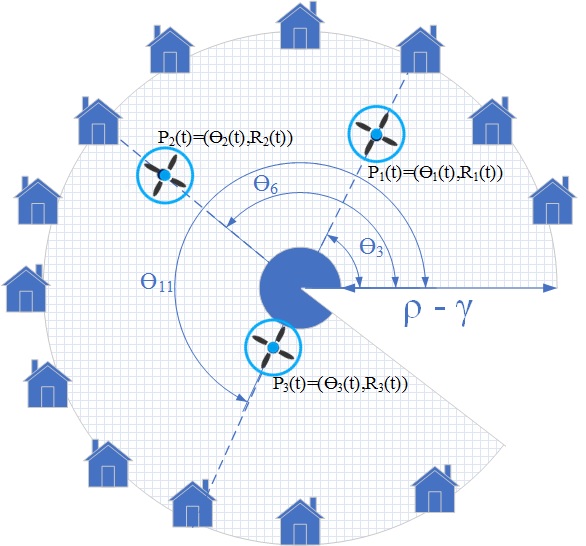}%
	\caption{First process trajectory}%
	\label{fig:idealcase}%
\end{figure}

The speed of drone $d$ at time $t$ is given by
\begin{equation} \label{eq:speedcase1}
V_{d}(t)=\begin{cases}
\frac{(\rho^2 -\gamma^2)}{2\sqrt{\tau((\rho^2 -\gamma^2)(t-k \tau-T_d)+\tau\gamma^2)}}, & \text{if $T_d + k \tau \leq t \leq T_d+(k+1) \tau$, k even}.\\
\frac{-(\rho^2 -\gamma^2)}{2\sqrt{\tau((\rho^2 -\gamma^2)((k+1)\tau+T_d-t)+\tau\gamma^2})}, & \text{if $T_d + k \tau \leq t \leq T_d+(k+1) \tau$, k odd}.
\end{cases}
\end{equation}

We prove that if $d_{th}$ drone flies with speed $V_d(t)$ at time $t$ given by (\ref{eq:speedcase1}), the drones will provide a uniform coverage over the area $A$. Equation (\ref{eq:speedcase1}) suggests that drones fly faster close to post office and decrease their speed near the boundary (i.e., near the houses) to provide a uniform coverage. Furthermore, the location of the drone is obtained by taking integral of (\ref{eq:speedcase1}) as in (\ref{eq:loccase1}).

\vspace{-5mm}

\begin{equation} \label{eq:loccase1}
R_{d}(t)=\begin{cases}
\sqrt{\frac{(\rho^2 -\gamma^2)(t-k \tau- T_d)}{\tau}+\gamma^2}, & \text{if $T_d + k \tau \leq t \leq T_d+(k+1) \tau$, k even}.\\
\sqrt{\frac{(\rho^2 -\gamma^2)((k+1)\tau+T_d-t)}{\tau}+\gamma^2}, & \text{if $T_d + k \tau \leq t \leq T_d+(k+1) \tau$, k odd}.
\end{cases}
\end{equation}

\begin{thm}\label{thm:case1}
	For trajectory process corresponding to the ideal case: i) For all $t > \tau$ , the instantaneous locations of the drones along the delivery path $(\theta_{d}(t), R_{d}(t))$, form a uniform BPP, and ii) the time to deliver $m$ packages is equal or less than $\frac{2m\tau}{D}+\tau$, i.e., $\Tau_m (A) \leq \frac{2m\tau}{D}+\tau$.
\end{thm}
Before providing the proof, we present the following lemma which will be used later in the proof procedure.
\begin{lem} \label{lem:CDF_uniformity}
	For any arbitrary observation time of $t > \tau$ , the location of any of the $D$ drones that move according to (\ref{eq:speedcase1}) has the following probability density function (pdf):
\[f_{R_d}(r_d)= \frac{2r_d}{\rho^2-\gamma^2}, \ \ \gamma \leq r_d \leq \rho. \]
That is, $f_{R_d}(r_d)$ is the pdf of distance of a uniformly distributed point in the circular region between radii $\gamma$ and $\rho$.
\end{lem}
\begin{proof}
	First, assume that $T_d+k\tau \leq t \leq T_d+(k+1)\tau$ and $k$ is odd, we have the following:
	
	\begin{equation}\label{eq:CDF_Td}
	\begin{aligned}
	F_{R_d}(r_d) = Pr(R_d(t) \leq r_d) = Pr( \sqrt{\frac{(\rho^2 -\gamma^2)((k+1)\tau+T_d-t)}{\tau}+\gamma^2} \leq r_d)\\= Pr(T_d \leq \frac{\tau(r_d^2-\gamma^2)}{\rho^2-\gamma^2}-(k+1)\tau+t)=Pr(T_d \leq \omega_d)= F_{T_d}(\omega_d),
	\end{aligned}
	\end{equation}
	where $F_{T_d}$ is the CDF of $T_d$ and $\omega_d=\frac{\tau(r_d^2-\gamma^2)}{\rho^2-\gamma^2}-(k+1)\tau+t$.\\
	Now to obtain the PDF of the $R_d$, we take the derivative of $F_{R_d}$:
	
	\begin{equation}\label{eq:PDF_uniformity}
	\begin{aligned}
	f_{R_d}(r_d)=\frac{dF_{R_d}(r_d)}{dr_d}=\frac{dF_{T_d}(\omega_d)}{dr_d}	=\frac{d}{dr_d}(\frac{r_d^2-\gamma^2}{\rho^2-\gamma^2}-(k+1)+\frac{t}{\tau})=\frac{2r_d}{\rho^2-\gamma^2}, \ \ \gamma \leq r_d \leq \rho
	\end{aligned}
	\end{equation}
	where (\ref{eq:PDF_uniformity}) is obtained from the fact that $T_d\sim U(0,\tau)$.  The case for even $k$ is proved similarly.
\end{proof}

We now provide the proof for Theorem \ref{thm:case1}.
\begin{proof}
	To prove the first part of Theorem \ref{thm:case1}, we first need to show that for $t > \tau$ , the location of vehicles are independent. This is intuitive, since $\theta_d\sim U(0, \theta_{max})$ and $T_d\sim U(0,\tau)$ both have been chosen independently. Second, we have to show that the locations are uniformly distributed over $A$. To do so, we note that since, $\theta_d\sim U(0, \theta_{max})$, the angle of the drone is uniformly distributed between $0$ and $\theta_{max}$, i.e. $\angle P_{d}(t)\sim U(0, \theta_{max})$. In addition, in Lemma \ref{lem:CDF_uniformity}, we proved that the location of drones, i.e., $R_{d}(t)$, are uniformly distributed over $A$. Therefore, drones are distributed according to uniform BPP over $A$.
	
	The proof of the second part of Theorem \ref{thm:case1} is as follows:  The departure times of the $D$ drones, $T_1, T_2, ..., T_D$, are i.i.d. and uniform over $(0, \tau)$. So by time $\tau$, all $D$ drones have departed and by time $3\tau$, they have delivered at least $D$ packages and come back to the post office center. The delivery time of the rest of packages (i.e., $m-D$ packages) is $\frac{2\tau(m-D)}{D}$, which are  simultaneously delivered by the $D$ drones. Therefore, the time to deliver $m$ packages is equal to or less than $\frac{2m\tau}{D}+\tau$.
\end{proof}

By considering the upper bound of $T_m(A)$ obtained in Theorem~\ref{thm:case1}, and the lower bound of delivery efficiency time obtained in Lemma \ref{lem:lem-low-Tau}, the efficiency of the proposed algorithm satisfies
	\begin{equation}\label{eq:efficiencyth1}
	\eta \geq \frac{1}{1+\frac{D}{2m}}.
	\end{equation}
Note that since $m$ is the number of delivered packages, the efficiency approaches $1$ over time.
\section{Practical (General) case}
\label{sec:Practical_case}

 In the general scenario, we do not want to impose specific assumptions on the density or location of homes or the distribution of arrival packages. Therefore, this setting can be applied to any neighborhood area.

\subsection{System Model and the Scenario}
 \begin{figure}[htbp]%
	\centering 	
	\includegraphics[width=0.6\columnwidth]{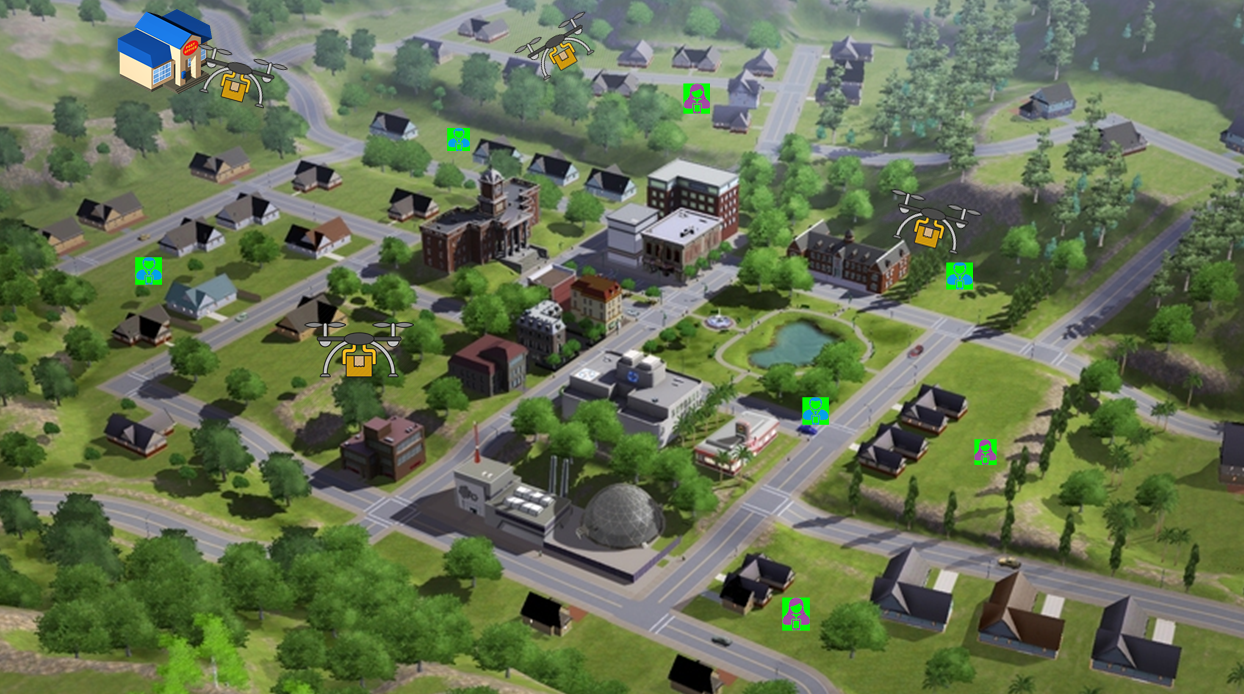} 	
	\caption{Neighborhood areas for Practical case}
	\label{fig:Practicalcasesneighborhood_area} 	
\end{figure}
Figure \ref{fig:Practicalcasesneighborhood_area} shows a typical neighborhood area over which we want to provide a uniform coverage. In this case, the geometry of neighborhood area does not need to be circular and is generally represented by a 2D shape. In addition, the houses are arbitrarily distributed in the neighborhood area, so the distances from the post office to the houses can be any arbitrary value. Again, we consider a multipurpose scenario: We assume that $D$ drones deliver the arriving packages from the post office to $N$ destination houses and at the same time we want to use them in a UCA framework. We assume packages are continuously arriving at the post office center. The only assumption we make (about the probability distributions of the destinations) is that over a period of time, each destination has non-zero probability.
%
The location of $h^{th}$ house is defined in a three-dimensional (3D) Cartesian coordinate system by $(x_h, y_h, 0)$,  where $1 \leq h \leq N$. Drones fly at a constant altitude $H$ above the ground and the location of the $d^{th}$ drone at time $t$ is shown by $(X_{d}(t), Y_{d}(t), H)$, where $1 \leq d \leq D$.

\vspace{-7mm}

\subsection{The Algorithm}
Here, we provide the detailed steps and components of the algorithm for the practical case.
\textbf{Division of the area to small cells:}
In this algorithm, referred to as Algorithm 2,  first, we divide the neighborhood area into small regions (cells). We use $A_l$ to refer to these regions where $1 \leq l \leq S$ and $S$ is the number of cells. We assume that $A_l$ is small so that at most one drone can fly over the cell at any time. This assumption is compatible with  the safety concern of drones as well.

\begin{algorithm}
	\begin{algorithmic}[1]
		\footnotesize	
		\Function{DeliveryCost}{$A, D, m, X$}		
		\Inputs{
			$A$ the area should be covered\\
			D drones with average speed \textit{V}\\
			$m$ number-of packages to deliver\\
			$X$ arrival packages which are not uniformly distributed over $1, 2, ..., N$}
		\Output{Total time to deliver $m$ packages ($T_m$)}
		\State Define $V_{MAX}$ and $V_{MIN}$
		\State Divide $A$ into small cells; called these cells $A_1, A_2, ..., A_S$
		\State for each small cells consider coverage probability $p_r\,,\, 1\,< r\,<S$ and initialize it with 0
		\For{\textit{h=1; $h<=N$}}
		\State Generate the straight trajectory between the post office and $h^{th}$ house and called it $PT_h$
		\EndFor
		\For{\textit{l=1; $l<=S$}}
		\If{No $PT$ passes through $A_l$}
		\State Select $PT_h$ which is the closest trajectory to $A_l$
		\State Change $PT_h$ in such a way that it passes  through $A_l$
		\EndIf
		\EndFor
		\For{\textit{j=1; $j<=\frac{m}{D}$}}
		\For{\textit{i=1; $i<=D$}}
		\State Assign $((j-1)*D+i)^{th}$ package to $i^{th}$ drone
		\State Assume $h$ is the destination of $((j-1)*D+i)^{th}$ package
		\State foreach region $l$ which $PT_h$ passes through		
		\If{$p_l\,<p^*$}
		\State Set velocity of $i^{th}$ drone to $MAX(V_{MIN},\frac{H_1(PT_h,A_l)}{p^*-p_l})$
		\Else
		\State Set velocity of $i^{th}$ drone to $MIN(V_{MAX},\frac{H_1(PT_h,A_l)}{p^*-p_l})$
		\EndIf
		\State Update $p_l$
		\EndFor
		\EndFor	
		\EndFunction
	\end{algorithmic}
	\caption{Algorithm corresponding to the practical case}
	\label{alg:case2}	
\end{algorithm}

\textbf{Defining Trajectories:} Then, we should define the trajectory paths, $PT_h:\,1\leq h\leq N$, between the post office and the houses in order to deliver the packages with high efficiency and simultaneously provide the uniform coverage. If we were not concerned about the UCA requirement, the most efficient trajectories would have been straight lines from post office to the destinations. Nevertheless, to achieve the UCA requirement, we might need to change trajectories slightly: If needed, we change the straight lines between the post office and the houses in a way that all defined small regions are crossed by at least one trajectory. It means that we want to make sure $((\cup_{h=1}^N{PT_h})\cap A_l\,\neq \varnothing)$ for any region $l$, $1 \leq l \leq S$.

\textbf{Uniform Coverage:} Here, we specifically state the requirement for uniform coverage. Consider the time interval $[0,t]$ where packages are continuously being delivered to their destinations. For any cell $l$, define $c_l(t)$ as the total time that cell is covered (i.e., a drone is flying over that region). The coverage ratio up to time $t$ is defined as $p_l(t)=\frac{c_l(t)}{t}$.
For uniform coverage, we require that for all cells $l=1,2,\cdots, S$, we must have $ \lim_{t \rightarrow \infty} p_l(t)=p^*,$
where $p^*$ is the desired coverage probability. It is worth noting that although to have a rigourous proof we state the condition for the limit case, in practice the convergence is fast as observed in our simulations in Section \ref{sec:result}.

\textbf{Varying Drone Speeds:}  Algorithm 2 is an \emph{adaptive} algorithm, that is, we adjust the velocity of drones when they enter the regions in order to preserve the uniformity in all cells. Intuitively, if the current coverage ratio is less than the desired coverage probability $p^*$ (i.e., $p_l(t)<p^*$), we should decrease the velocity of the drone, and if it is more than the expected coverage probability, the drone should pass this region faster. Lines 21 to 25 of Algorithm~\ref{alg:case2} show this adjustment, where $H_1$ is  Hausdorff measure.

\begin{figure}[htbp]%
	\centering 	
	\includegraphics[width=0.7\columnwidth]{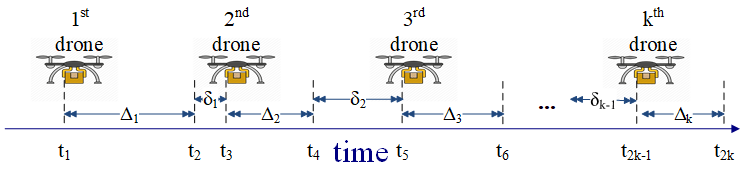} 	
	\caption{Arrival/depature of drones over time within a cell }
	\label{fig:time_axis} 	
\end{figure}

Below we show that we can adjust the velocities in a way to guarantee $ \lim_{t \rightarrow \infty} p_l(t)=p^*$ for all cells, $A_l$, $l=1,2,\cdots, S$. For notational simplicity, we will sometimes drop the subscript $l$ in the rest of the proof. Define $L$ as $H_1((\cup_{h=1}^N{PT_h})\cap A_l)$, i.e., the lengths of the part of trajectories restricted to cell $l$. Fig.~\ref{fig:time_axis} demonstrates arrival/departure of drones over the region during a delivery period for $m$ packages. As you can see, first the drone arrives over the region at time $t_1$ and traverses the cell with speed $V_1$, and leaves the region at time $t_2$. In general, the  $k^{th}$ drone arrives over the cell at time $t_{2k-1}$ and leaves the region at time $t_{2k}$ (traverses the cell with speed $V_k$). The time between the arrival and departure of the $k^{th}$ drone in cell $l$ is denoted by $\Delta_k$ and the time between departure of the $k^{th}$ drone and arrival of the $(k+1)^{th}$ is shown by $\delta_k$. Thus,
\begin{align*}
\Delta_k=t_{2k}-t_{2k-1},\text{and }\delta_k=t_{2k+1}-t_{2k}.
\end{align*}
Suppose the maximum and minimum possible speeds of drones are given by $V_{MAX}$ and $V_{MIN}$. If we define $\Delta_{MAX}=\frac{L}{V_{MIN}}$ and $\Delta_{MIN}=\frac{L}{V_{MAX}}$, then we have $0<\Delta_{MIN} \leq \Delta_k \leq \Delta_{MAX}$. In any practical scenarios, the $\delta_k$ values can not be unlimited. So here we assume that there exist $\delta_{MIN} \ge 0$ and $\delta_{MAX} \ge 0$ such that for all $k$, $0<\delta_{MIN} \leq \delta_k \leq \delta_{MAX}$. Before stating and proving the main theorem, we need the following definition:

	\begin{define} (Causal Velocity Profiles)
		An algorithm for determining $V_j$ for $j=1, 2, \dots$ is said to be casual, if the value of $V_j$ is determined only by the past data up-to time $t_{2j-1}$.
	\end{define}

\begin{thm}\label{thm:case2}
If $V_{MIN}$ and $V_{MAX}$ can be chosen such that $V_{MAX} \geq \frac{L(1-p^*)}{p^*\delta_{MIN}}$ and $V_{MIN} \leq \frac{L(1-p^*)}{p^*\delta_{MAX}}$, then there exists a causal velocity profile such that $\lim_{t\to\infty}p_l(t)=p^*$.\footnote{This theorem is a main result stating that the UCA requirement can be satisfied. The conditions $V_{MIN}$ and $V_{MAX}$ simply state that we should be able to have a large enough range for the velocities to be able to achieve a uniform coverage. The proof is given below, which is a bit technical due to the fact that we want to prove the statement in a very general scenario without making specific assumptions. The readers less interested in the technical proof, can refer to Section \ref{sec:result} to see the simulation results showing the performance of the proposed algorithms for two real neighborhood areas: University of Massachusetts Amherst and Union Point, which is a smart city near Boston.}
\end{thm}
Before proving this theorem, we provide some lemmas that are later used during the proof. Also for simplicity, we assume only one path goes through $A_l$ (The proof can easily be extended to multiple paths). Since at any time, at most one drone flies over each cell, we can say that $c_l(t_{2k})=\Sigma_{i=1}^k\Delta_k$ and also $c_l(t_{2k+1})=\Sigma_{i=1}^k\Delta_i$. From these expressions, the following equations can be concluded:
	\vspace{-7mm}
	\begin{equation}
	p(t_{2k})=\frac{c(t_{2k})}{t_{2k}}=\frac{\Sigma_{i=1}^k\Delta_i}{t_{1}+\Sigma_{i=1}^k\Delta_i+\Sigma_{i=1}^{k-1}\delta_j}.
	\end{equation}
	\begin{equation}
	p(t_{2k+1})=\frac{c(t_{2k+1})}{t_{2k+1}}=\frac{\Sigma_{i=1}^k\Delta_i}{t_{1}+\Sigma_{i=1}^k(\Delta_i+\delta_i)}.
	\end{equation}
	
	\begin{lem} \label{lem:timegoesinfinity}
		If all drones traverse the cell with maximum speed at any time i.e. $V_j=V_{MAX}$ for all $j$, then $\limsup\limits_{k\rightarrow\infty}p(t_{2k})\leq p^*$ and if all drones traverse the cell with minimum speed at any time i.e. $V_j=V_{MIN}$ for all $j$, then $\liminf\limits_{k\rightarrow\infty}p(t_{2k+1})\geq p^*$.
	\end{lem}
	\begin{proof}[Proof of Lemma \ref{lem:timegoesinfinity}]
		If all drones pass the cell with the maximum velocity  $V_{MAX}$, it takes $\frac{L}{V_{MAX}}$ to leave the cell and we can obtain the probability coverage as follows:
		\begin{align*}
		P(t_{2k+1})=\frac{\Sigma_{i=1}^k\Delta_i}{t_{1}+\Sigma_{i=1}^k\Delta_i+\Sigma_{i=1}^{k}\delta_j}=\frac{\frac{kL}{V_{MAX}}}{t_1+\frac{kL}{V_{MAX}}+\Sigma_{i=1}^{k}\delta_j}=\frac{1}{1+\frac{t_1V_{MAX}}{kL}+\frac{V_{MAX}}{kL}\Sigma_{i=1}^{k}\delta_j}.
		\end{align*}
		By using $V_{MAX} \geq \frac{L(1-p^*)}{p^*\delta_{MIN}}$, we have
		\begin{align*}
	     \frac{1}{1+\frac{t_1V_{MAX}}{kL}+\frac{V_{MAX}}{kL}\Sigma_{i=1}^{k}\delta_j} \leq \frac{1}{1+\frac{t_1V_{MAX}}{kL}+\frac{1-p^*}{p^*\delta_{MIN}}(\frac{1}{k}\Sigma_{i=1}^{k}\delta_j)},
		\end{align*}
		and since $(\frac{1}{k}\Sigma_{i=1}^{k}\delta_j) \geq \delta_{MIN}$, we have
		\begin{align*}
		\limsup\limits_{k\rightarrow\infty}p(t_{2k+1}) \leq \frac{1}{1+\frac{1-p^*}{p^*}1},
		\end{align*}
		and as a result,
		$\limsup\limits_{k\rightarrow\infty}p(t_{2k+1}) \leq p^*.$
		
		Next, we show if all drones traverse the cell with minimum speed at any time i.e. $V_j=V_{MIN}$ for all $j$, then $\liminf\limits_{k\rightarrow\infty}p(t_{2k+1})\geq p^*$. In this case,
		\begin{align*}
		P(t_{2k+1})=\frac{\Sigma_{i=1}^k\Delta_i}{t_{1}+\Sigma_{i=1}^k\Delta_i+\Sigma_{i=1}^{k}\delta_j}=\frac{\frac{kL}{V_{MIN}}}{t_1+\frac{kL}{V_{MIN}}+\Sigma_{i=1}^{k}\delta_j}=\frac{1}{1+\frac{t_1V_{MIN}}{kL}+\frac{V_{MIN}}{kL}\Sigma_{i=1}^{k}\delta_j}.
		\end{align*}
		By using $V_{MIN} \leq \frac{L(1-p^*)}{p^*\delta_{MAX}}$, we have
		\begin{align*}
		\frac{1}{1+\frac{t_1V_{MIN}}{kL}+\frac{V_{MIN}}{kL}\Sigma_{i=1}^{k}\delta_j} \geq \frac{1}{1+\frac{t_1V_{MIN}}{kL}+\frac{1-p^*}{p^*\delta_{MAX}}(\frac{1}{k}\Sigma_{i=1}^{k}\delta_j)}=\frac{1}{1+\frac{t_1V_{MIN}}{kL}+\frac{1-p^*}{p^*}(\frac{1}{k\delta_{MAX}}\Sigma_{i=1}^{k-1}\delta_j)}.
		\end{align*}
		and since $(\frac{1}{k}\Sigma_{i=1}^{k}\delta_j) \leq \delta_{MAX}$, we have
		\begin{align*}
		\liminf\limits_{k\rightarrow\infty}p(t_{2k+1}) \geq \frac{1}{1+\frac{1-p^*}{p^*}1},
		\end{align*}
		and as a result,
		$
		\liminf\limits_{k\rightarrow\infty}p(t_{2k+1}) \geq p^*.
		$
	\end{proof}
    Note that the above argument can be repeated for the cases where the first $k_0$ values of $V_j$'s are arbitrary as they do not impact the limiting behavior. So we can provide the following corollary.
	\begin{cor}
	Let $k_0$ be a positive integer. If we have a sequence $V_j$ for $j = 1, 2,\dots , \infty$ such that for all $j \ge k_0$, $V_j = V_{MAX}$, then
	\renewcommand{\theequation}{\Alph{equation}}
	\setcounter{equation}{0}
	\vspace{-7mm}
	\begin{equation}\label{eq:A}
	\begin{aligned}
	\limsup\limits_{k\rightarrow\infty}p(t_{2k+1}) \leq p^*.
	\end{aligned}
	\end{equation}
	Similarly, If we have a sequence $V_j$ for $j = 1, 2,\dots , \infty$ such that for all $j \ge k_0$, $V_j = V_{MIN}$, then
	\vspace{-7mm}
	\begin{equation}\label{eq:B}
	\begin{aligned}
	\liminf\limits_{k\rightarrow\infty}p(t_{2k+1})\geq p^*.
	\end{aligned}
	\end{equation}
	\end{cor}
    For the brevity of the notation, let's define $p(t_{2k+1}, V_{MIN})$ as the value of $p(t_{2k+1})$ when for all $j \ge k_0$, $V_j = V_{MIN}$, and define $p(t_{2k+1}, V_{MAX})$, similarly. Thus, we have
    \[ \liminf\limits_{k\rightarrow\infty}p(t_{2k+1}, V_{MIN})\geq p^*.\]
    Now consider two cases: If we have
    \[ \limsup\limits_{k\rightarrow\infty}p(t_{2k+1}, V_{MIN})\leq p^*,\]
    Then, we will have
    \[ \lim \limits_{k\rightarrow\infty}p(t_{2k+1}, V_{MIN})= p^*.\]
    Otherwise, we must have
    \[ \limsup\limits_{k\rightarrow\infty}p(t_{2k+1}, V_{MIN})> p^*.\]
    So we come up with the following corollaries:
    \begin{cor} \label{cor:V_MIN}
    For any sequence of $p(t_{2k+1})$, one of the following is true:
    \[ \lim \limits_{k\rightarrow\infty}p(t_{2k+1}, V_{MIN})= p^*,\text{or }
 \limsup\limits_{k\rightarrow\infty}p(t_{2k+1}, V_{MIN})> p^*.\]
    \end{cor}
     \begin{cor} \label{cor:V_MAX}
    For any sequence of $p(t_{2k+1})$, one of the following is true:
    \[ \lim \limits_{k\rightarrow\infty}p(t_{2k+1}, V_{MAX})= p^*\text{, or }
 \liminf\limits_{k\rightarrow\infty}p(t_{2k+1}, V_{MAX})< p^*.\]
    \end{cor}

	\begin{define} (Min-Max Algorithm)
		The min-max algorithm for choosing $V_i$'s is defined as follows: We choose $V_1=V_{MIN}$. For $k \ge 1$, if $p(t_{2k-1})\leq p^*$, then $V_{k}=V_{MIN}$, otherwise $V_{k}=V_{MAX}$.
	\end{define}
\emph{Note:} The min-max algorithm is used below to prove Theorem \ref{thm:case2}. Nevertheless, there are various choices of velocity profiles $V_j,j=1, 2,\dots,\infty$ that satisfy Theorem \ref{thm:case2}. Their differences are in their rate of convergence and their practicality. The one we have chosen in our algorithm provides a very fast convergence (Algorithm \ref{alg:case2}) and also results in much smoother operation (the changes in speeds can actually be made minimal and gradual suitable for practical implementation). However, for the sake of proofs, it is easier to use the min-max algorithm defined above.
\begin{lem} \label{lem:updown}
For the min-max algorithm, the following statements are true:
\begin{enumerate}
  \item If $p(t_{2k-1})<p^*$, then $p(t_{2k+1}) \geq p(t_{2k-1})$.
  \item If $p(t_{2k-1})>p^*$, then $p(t_{2k+1}) \leq p(t_{2k-1})$.
\end{enumerate}
\end{lem}

\begin{proof}
If $p(t_{2k-1})<p^*$, then $V_{k}=V_{MIN}$, so $\Delta_k=\Delta_{MAX}=\frac{L}{V_{MIN}}$.
	\renewcommand{\theequation}{\arabic{equation}}
	\setcounter{equation}{9}
  	\begin{equation}\label{eq:tk+-1}
	p(t_{2k+1})=\frac{c(t_{2k+1})}{t_{2k+1}}=\frac{c(t_{2k-1})+\Delta_{MAX}}{t_{2k-1}+\Delta_{MAX}+\delta_k}
	\end{equation}
     Now note that
     	\begin{equation} \label{eq:DMAXdMAX}
	\frac{\Delta_{MAX}}{\Delta_{MAX}+\delta_k} \geq  \frac{\Delta_{MAX}}{\Delta_{MAX}+\delta_{MAX}} \geq p^*.
	\end{equation}
     The last inequality is the direct result of the main assumption $V_{MIN} \leq \frac{L(1-p^*)}{p^*\delta_{MAX}}$.
     Now by combining $p(t_{2k-1})=\frac{c(t_{2k-1})}{t_{2k-1}}<p^*$ and Equations \ref{eq:tk+-1} and \ref{eq:DMAXdMAX}, we conclude $p(t_{2k+1}) \geq p(t_{2k-1})$. The second statement of the lemma can be proved similarly.
\end{proof}

\begin{lem} \label{lem:tj+1tj}
For the min-max algorithm, we have
$\lim\limits_{j\rightarrow\infty}\left| p(t_{j+1}) - p(t_{j}) \right| = 0$.
\end{lem}

\begin{proof}
It suffices to show $\lim\limits_{k\rightarrow\infty}\left| p(t_{2k}) - p(t_{2k+1}) \right| = 0$ and $\lim\limits_{k\rightarrow\infty}\left| p(t_{2k}) - p(t_{2k-1}) \right| = 0$. The proofs are similar, so we just show the first one. Recall that $P(t_{2k})=\frac{\Sigma_{i=1}^k\Delta_i}{t_{1}+\Sigma_{i=1}^k\Delta_i+\Sigma_{j=1}^{k-1}\delta_j}=\frac{U_k}{W_k}$. Thus, we have $p(t_{2k+1})=\frac{U_k}{W_k+\delta_k}$.
         Remember, that for $\delta_{MIN} \ge 0$ and $\delta_{MAX} \ge 0$, we have for all $k$, $0<\delta_{MIN} \leq \delta_k \leq \delta_{MAX}$ and $0<\Delta_{MIN} \leq \Delta_k \leq \Delta_{MAX}$. We have
        \begin{align*}
       	\begin{dcases*}
		k\Delta_{MIN} \leq U_k \leq k\Delta_{MAX}\\
		t_1+ k\Delta_{MIN} + (k-1) \delta_{MIN} \leq W_k \leq t_1+ k\Delta_{MAX} + (k-1) \delta_{MAX}.
		\end{dcases*}
		\end{align*}
Thus, $\lim_{k \rightarrow \infty} U_k=\infty$ and $\lim_{k \rightarrow \infty} W_k=\infty$, and their ratio $\frac{W_k}{U_k}$ is bounded. Therefore, we can conclude that
\vspace{-7mm}
		\begin{align*}
		\left|p(t_{2k})-p(t_{2(k+1)})\right| = \frac{\delta_k U_k}{(W_k+\delta_k)(W_k)} \to 0
		\end{align*}
		as $k$ goes to infinity.
\end{proof}

	\begin{lem} \label{lem:velocityprogfileexits}
		There exists a casual algorithm $V_j$ for $j = 1, 2,\dots , \infty$ such that
		
		\begin{equation}\label{eq:lemma3}
		\begin{aligned}
		\lim\limits_{j\rightarrow\infty}p(t_{j})= p^*.
		\end{aligned}
		\end{equation}		
	\end{lem}
	\begin{proof}
        Based on Lemma \ref{lem:tj+1tj}, it suffices to show $\lim\limits_{k\rightarrow\infty}p(t_{2k+1})= p^*$. We claim that using the min-max velocity profile, we can achieve $\lim\limits_{k\rightarrow\infty}p(t_{2k+1})= p^*$. Let $p(t_3) \leq p^*$, then the min-max algorithm adjusts $V_2$ to $V_{MIN}$. In fact, $V_{j+1}$ is tuned to $V_{MIN}$ as long as $p(t_{2j+1}) \leq p^*$. Now, if for all $j>1$, $p(t_{2j+1}, V_{MIN}) \leq p^*$ then by Corollary \ref{cor:V_MIN}, we have $\lim \limits_{k\rightarrow\infty}p(t_{2k+1}, V_{MIN})= p^*$, in which case we are done. Otherwise, there exists a $k_1$ in which $p(t_{2k_1+1}) \geq p^*$ at which point the algorithm switches to $V_{MAX}$. Similarly, by Corollary \ref{cor:V_MAX}, there exists $k_2 \geq k_1$ such that $p(t_{2k_2+1}) \leq p^*$, and this oscillation repeats infinitely (or anytime it stops we are already converging to $p^{*}$ and we are done). Thus, we may assume the sequence $p(t_j)$, for $j= 1, 2, \dots$ \textbf{crosses} $p^*$ infinitely many times.

		To complete the proof of Lemma \ref{lem:velocityprogfileexits}, we show that for all $\epsilon > 0$, there exists $k_\epsilon$ such that for all $ k > k_\epsilon$, we have $\left| p(t_{2k+1})-p^* \right| < \epsilon$. First, choose $k_1$ such that for all $k \geq k_1$, we have $ \left| p(t_{2k})-p(t_{2(k-1)}) \right| < \frac{\epsilon}{4}$ and $ \left| p(t_{2k+1})-p(t_{2k-1}) \right| < \frac{\epsilon}{4}$ (Lemma \ref{lem:tj+1tj}).

Without loss of generality assume $p(t_{2k-1})<p^*$. Let $k_\epsilon$ be the smallest $k > k_1$ such that $p(t_{2k_\epsilon+1})$ crosses $p^*$, then we know the following
\begin{enumerate}
  \item $p(t_{2k_\epsilon+1})>p^*$ and $p(t_{2k_\epsilon+3})<p(t_{2k_\epsilon+1})$ (Lemma \ref{lem:updown});
  \item $|p(t_{2k_\epsilon+1})-p^*|<\frac{\epsilon}{2}$;
  \item $|p(t_{2k_\epsilon+3})-p(t_{2k_\epsilon+1})|<\frac{\epsilon}{2}$.
\end{enumerate}
Therefore, we conclude $|p(t_{2k_\epsilon+3})-p^{*}|<\epsilon$. Indeed, repeating the same argument from now on, we conclude for all $ k > k_\epsilon$, we have $\left| p(t_{2k+1})-p^* \right| < \epsilon$.
\end{proof}
\begin{proof}[Proof of Theorem \ref{thm:case2}]	To prove Theorem \ref{thm:case2}, we show that the min-max sequence $V_j$ satisfies $\lim_{t\to\infty}p_l(t)=p^*$. It should be noted that we assumed that there exists $\delta_{MIN} > 0$ and $\delta_{MAX} > 0$ such that for all $k$, $0 < \delta_{MIN} \leq \delta_k \leq \delta_{MAX} < \infty $, also there are $\Delta_{MIN} > 0$ and $\Delta_{MAX} > 0$ such that for all $k$, $0 < \Delta_{MIN} \leq \Delta_k \leq \Delta_{MAX} < \infty$.
	
	Here, we define $k(t) =min \left(k: t_{2k}\geq t \right)  $. Also, for $t_{2(k-1)} \leq t \leq t_{2k}$ we define the following:
	\begin{align*}
	a_k=\frac{c(t_{2k})}{t_{2k}-\delta_{MAX}-\Delta_{MAX}},\,
	b_k=\frac{c(t_{2(k-1)})}{t_{2(k-1)}+\delta_{MAX}+\Delta_{MAX}}
	\end{align*}
	By using (\ref{eq:lemma3}), we have
	\begin{align*}
	\lim\limits_{k\rightarrow\infty}a_k=\lim\limits_{k\rightarrow\infty}\frac{c(t_{2k})}{t_{2k}-\delta_{MAX}-\Delta_{MAX}}=\lim\limits_{k\rightarrow\infty}\frac{c(t_{2k})}{t_{2k}}\frac{t_{2k}}{t_{2k}-\delta_{MAX}-\Delta_{MAX}}=p^*.	
	\end{align*}
	Similarly, we can conclude that
	\begin{align*}
	\lim\limits_{k\rightarrow\infty}b_k=\lim\limits_{k\rightarrow\infty}\frac{c(t_{2(k-1)})}{t_{2(k-1)}+\delta_{MAX}+\Delta_{MAX}}=\lim\limits_{k\rightarrow\infty}\frac{c(t_{2(k-1)})}{t_{2(k-1)}}\frac{t_{2(k-1)}}{t_{2(k-1)}+\delta_{MAX}+\Delta_{MAX}}=p^*
	\end{align*}
	Using definition of $k(t)$, we have
	\begin{align*}
	p(t) = \frac{c(t)}{t} \leq \frac{c(t_{2k})}{t} \leq \frac{c(t_{2k}))}{t-\delta_{MAX}-\Delta_{MAX}} = a_k\\
	p(t) = \frac{c(t)}{t} \geq \frac{c(t_{2(k-1)})}{t} \geq \frac{c(t_{2(k-1)}))}{t+\delta_{MAX}+\Delta_{MAX}} = b_k
	\end{align*}
	So, for all $t$, we have $b_k(t) \leq p(t) \leq a_k(t)$. Based on this  we can conclude that:
	\begin{align*}
	\begin{dcases*}
	p(t) \geq b_k(t)\\
	p(t) \leq a_k(t)
	\end{dcases*}\Rightarrow\begin{dcases*}
	\liminf\limits_{t\rightarrow\infty}p(t) \geq \liminf\limits_{t\rightarrow\infty}b_k(t) = p^*\\
	\limsup\limits_{t\rightarrow\infty}p(t) \leq \limsup\limits_{t\rightarrow\infty}a_k(t) = p^*
	\end{dcases*}\Rightarrow \lim\limits_{t\rightarrow\infty}p(t)= p^*
	\end{align*}
\end{proof}



\section{Simulation Results}
\label{sec:result}
In Section \ref{sec:Ideal_case}, we proved that the ideal algorithm provides uniform coverage, in this section, we run simulation for this algorithm to verify our claim. As mentioned before, to investigate the coverage associated to each trajectory, we divide the neighborhood area into small cells and measure the average number of drones over the regions through simulation. We consider $10$ disjoint equal cells within $\frac{5}{8}$ of a circular area with radius $\rho = 5 km$ as shown in Fig~\ref{fig:disjiont_equal_area}. We set the radius of the post office center to $100$, i.e., $\gamma = 100 m$, and the number of houses to $100$, i.e., $N = 100$. We run the simulation with two different number of drones $D = 5$ and $D = 10$.



\begin{figure}[htbp]%
	\centering
	\includegraphics[width=0.4\columnwidth]{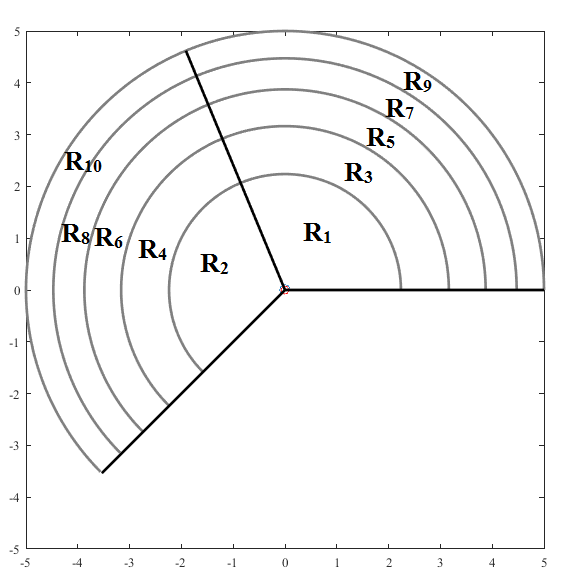}%
	\caption{Circular area with radius 5 km is divided to 10 disjoint regions}%
	\label{fig:disjiont_equal_area}%
\end{figure}

Figure~\ref{fig:uniform_c1} shows the average number of drones flying over each of the ten regions for both simulation and analysis. As can be seen, the simulation and analysis results coincide for both 5 and 10 drones. Also, there is an equal average number of drones over all the regions, which validates our claim that the proposed Algorithm 1 provides uniform coverage.

\begin{figure}[htbp]%
	\centering	
	\includegraphics[width=0.9\columnwidth]{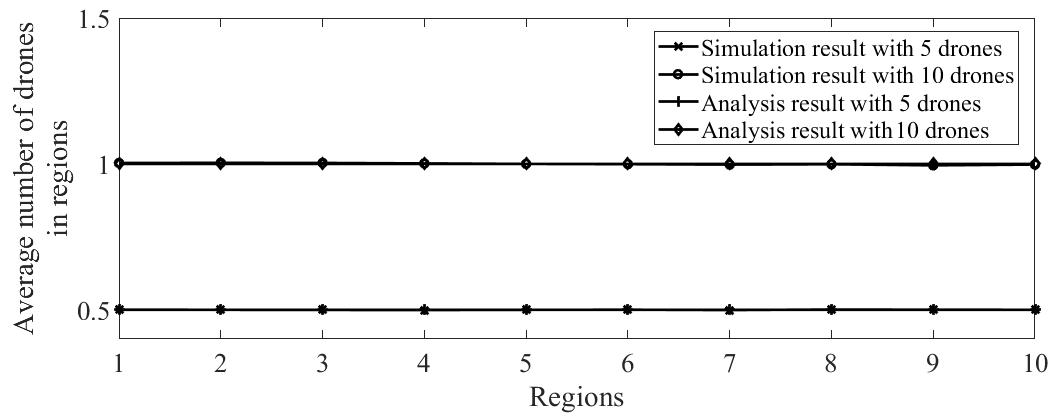}
	\caption{ Average number of the drones over the regions for 5 and 10 drones }
	\label{fig:uniform_c1}%
\end{figure}

In Section \ref{sec:Practical_case}, we proposed Algorithm 2 to deliver the packages and provide the uniform coverage simultaneously which can be applied to any neighborhood area with any distribution of arrival packages and position of houses. We consider two neighborhood areas, University of Massachusetts Amherst and Union Point,which is a smart town near Boston, to verify our claim about uniformity in coverage and investigate the efficiency of our algorithm to deliver the packages. We introduced University of Massachusetts Amherst community in Section \ref{sec:intro}. Figure \ref{fig:umassmap} and \ref{fig:heatmapumass} showed the neighborhood map and the heat-map of average number of drones for the benchmark algorithm, respectively. Figure \ref{fig:umassproposed} shows the heat-map of the average number of drones for the proposed algorithm. In Figure \ref{fig:umassp5drones}, our proposed algorithm is simulated by 5 drones and in Figure \ref{fig:umassp10drones}, our algorithm is simulated by 10 drones. As can be seen, both configurations provide uniform coverage over the neighborhood area. 

\begin{figure}[htbp]%
	\centering
	\subfloat[5 drones are used to simulate the algorithm]{
		\includegraphics[width=0.4\columnwidth]{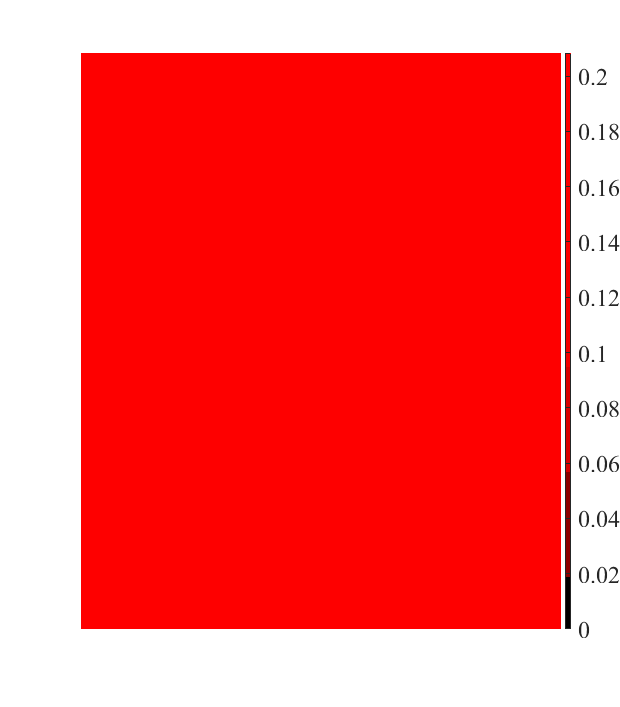}\label{fig:umassp5drones}}\hfill
	\subfloat[10 drones are used to simulate the algorithm]{
		\includegraphics[width=0.4\columnwidth]{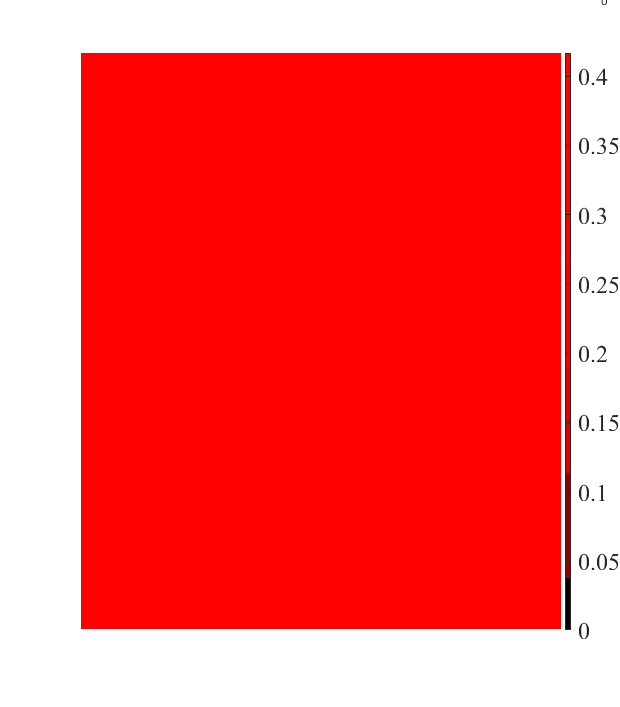}\label{fig:umassp10drones}}
	\caption{Proposed multi-purpose drone algorithm for University of Massachusetts (UMASS) community}%
	\label{fig:umassproposed}%
\end{figure}

As for the Union Point, which has approximately 4000 homes\cite{unionpoint} and total area of 1500 acres (see Fig. \ref{fig:uinionpointmap}), we assume the last-mile delivery office is located in the top-left corner of the figure. We divided the neighborhood community into 24 small cells to investigated the coverage. 10 drones are used to deliver the packages. First we assume the drones fly in straight lines with constant velocity to deliver the packages to  houses. The average number of drones flown over the regions is shown by a heat-map in Fig. \ref{fig:heatmapalg1unionpoint}. Then we assume the drones follow the proposed Algorithm 2 to deliver the packages to houses. The average number of drones over the regions is shown by heat-map in Fig. \ref{fig:heatmapproposedalgunionpoint}. As can be seen, the proposed algorithm provides uniform coverage over the entire neighborhood area.

\begin{figure}[htbp]%
	\centering
	\subfloat[Union point]{
		\includegraphics[width=0.32\columnwidth]{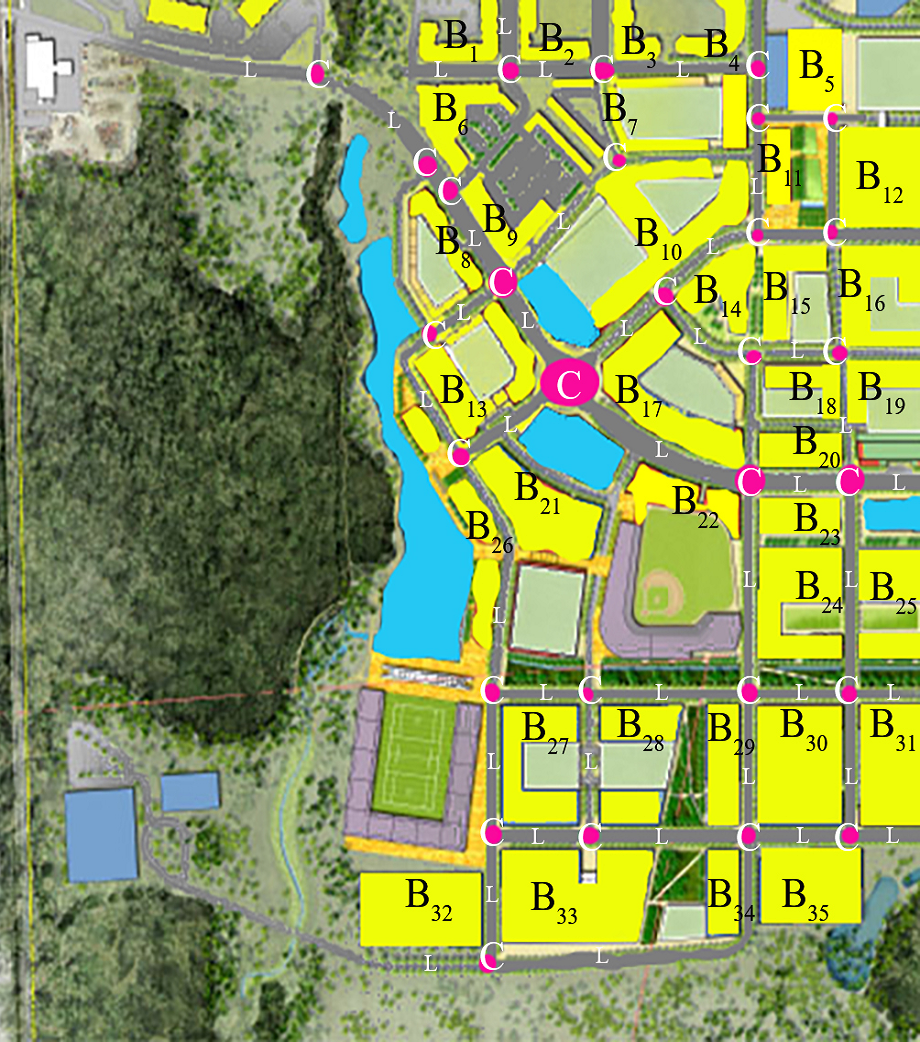}\label{fig:uinionpointmap}}\hfill
	\subfloat[Heat-map of average number of drones for the \textit{fixed-speed-direct-line} algorithm]{
		\includegraphics[width=0.32\columnwidth]{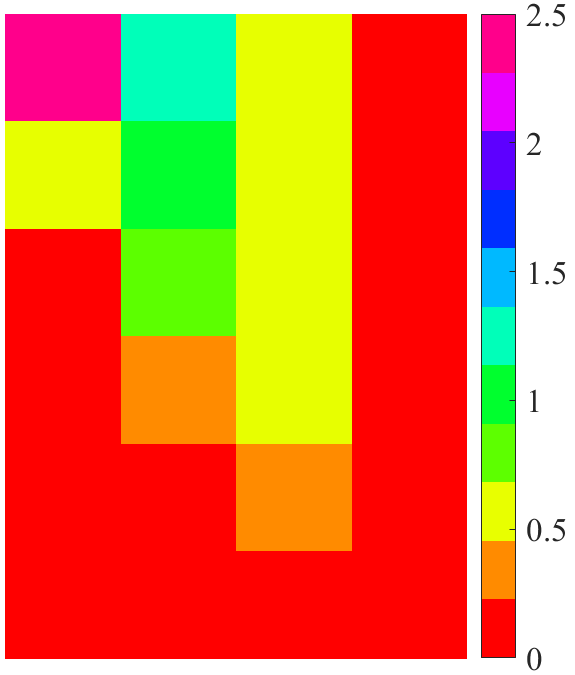}\label{fig:heatmapalg1unionpoint}}\hfill
	\subfloat[Heat-map of average number of drones for the \textit{proposed} algorithm]{
		\includegraphics[width=0.32\columnwidth]{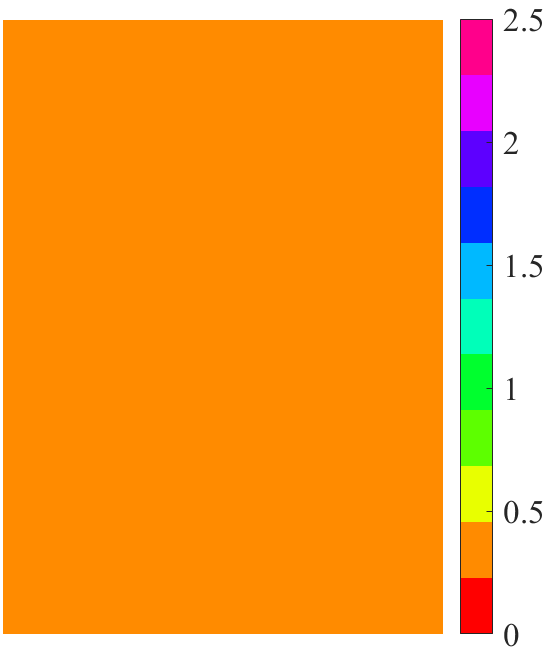}\label{fig:heatmapproposedalgunionpoint}}
	\caption{Multi-purpose drone algorithm for Union point community}%
	\label{fig:unionpointsim}%
\end{figure}

So far, we have shown that the proposed algorithm provides uniform coverage over the neighborhood area. Here, we want to show that this algorithm also provides efficient delivery of packages. To do so, we measured the average delivery time of 1000 packages through simulation and showed the efficiency in Table~\ref{tab:efficiency}. As seen, our proposed algorithm delivers the packages over both communities efficiently. In Union Point, the efficiency slightly decreases because there are some cells without buildings. Figure~\ref{fig:unionpointsim} shows the distribution of package delivery time for the Union Point community. As can be seen, the distribution profiles are of similar nature for the proposed algorithm and the bench mark algorithm while the latter can not provide a uniform coverage. In particular, we are interested in the fraction packages that are delivered later than certain amount of time, e.g., 30 minutes. This value has been reported 
 in Table~\ref{tab:efficiency}. As reported in this table, the fraction of these packages are very small. 

\begin{table}[htbp]
\centering
\caption{ Average time to deliver 1000 packages with 10 drones for second algorithm}
	\begin{tabular}{c|c|c|}
		\cline{2-3}
		& Efficiency & fraction of  packages (average) with delivery time \textgreater 30 mins \\ \hline
		\multicolumn{1}{|c|}{UMASS Community} & 1 & 0.006 \\ \hline
		\multicolumn{1}{|c|}{Union point Community} & 0.87 & 0.012 \\ \hline
	\end{tabular}
\label{tab:efficiency}
\end{table}

\begin{figure}[htbp]%
	\centering
	
		\includegraphics[width=0.6\columnwidth]{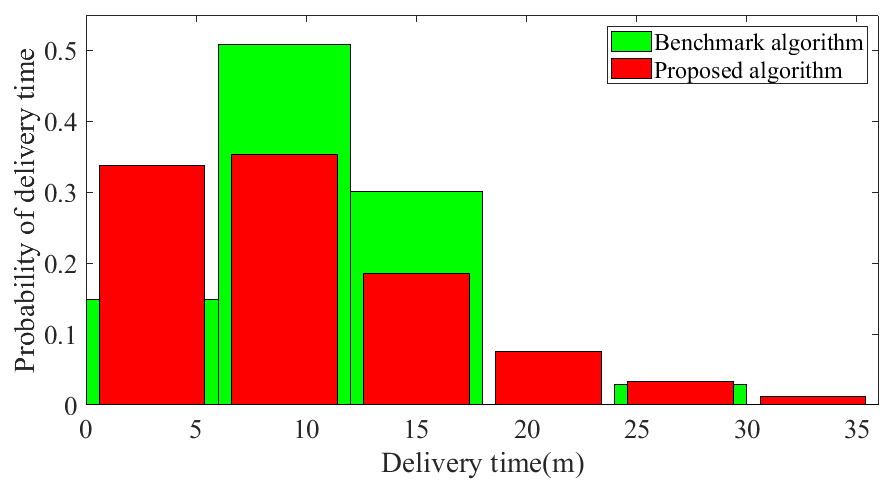}\label{fig:unionpountspmf}
	\caption{Probability of delivery time for Union point community}%
	\label{fig:unionpointsim}%
\end{figure}

\vspace{-2mm}
\section{Conclusion}
\label{sec:conclude}
In this paper, we proposed UAVs that simultaneously perform multiple tasks, uniform-coverage applications (UCAs) and last-mile delivery. We investigated the multi-task UAVs for two scenarios: i) a simplified scenario where the neighborhood area is a circular region, and ii) a practical scenario where the neighborhood area is an arbitrarily-shaped region. For each scenario, we proposed an algorithm for UCA and last-mile delivery. We proved that both algorithms provide a uniform coverage probability for a typical user within the neighborhood area. Through simulation results we verified the uniform coverage and at the same time, we demonstrated that we can still maintain the delivery efficiency compared to the original delivery algorithm.


\ifanonymous
{ }
\else

\fi



\bibliographystyle{IEEEtran}
\bibliography{bib/refs}

\end{document}